\documentclass[letterpaper, 10 pt, conference]{ieeeconf}  %

\IEEEoverridecommandlockouts                              %
\usepackage[compress]{cite}
\usepackage[hidelinks]{hyperref}

\usepackage{graphicx}

\title{\LARGE \bf
A Weak Notion of Symmetry for Dynamical Systems
}

\author{Jake Welde and Pieter van Goor%
\thanks{J. Welde is with the Sibley School of Mechanical and Aerospace Engineering, Cornell University, Ithaca, NY,  USA,
        {\tt\small jakewelde@cornell.edu}.
        P. van Goor is with the School
of Aerospace, Mechanical, and Mechatronic Engineering, The University
of Sydney, NSW, Australia, 
        {\tt\small pieter.vangoor@sydney.edu.au}.}%
}

\usepackage{amsmath, amssymb}

\newcommand{\Left}[0]{{\mathrm{L}}}
\newcommand{\Right}[0]{{\mathrm{R}}}

\DeclareMathOperator{\Diff}{Diff}
\DeclareMathOperator{\Aut}{Aut}

\DeclareMathOperator{\id}{id}

\newcommand{\aut}[0]{{\mathfrak{aut}}}
\newcommand{\X}[0]{{\mathfrak{X}}}
\newcommand{\flow}[0]{{\tau}{}}
\newcommand{\diff}[0]{{\mathrm{d}}}
\newcommand{\ddt}[0]{\tfrac{\diff}{\diff t}}
\newcommand{\T}{\mathrm{T}}

\usepackage{amsthm}
\usepackage{thmtools}
\renewcommand\thmcontinues[1]{\textit{continued}}

\declaretheoremstyle[notefont=\normalfont\itshape,bodyfont=\normalfont]{normaltext}

\newtheorem{theorem}{Theorem}
\newtheorem{corollary}{Corollary}
\declaretheorem[name=Definition,style=normaltext]{definition}
\declaretheorem[name=Remark,style=normaltext]{remark}

\newtheorem{proposition}{Proposition}
\newtheorem{lemma}{Lemma}
\newtheorem{fact}{Fact}

\usepackage{xcolor}

\usepackage[compress]{cite}

\begin{document}

\maketitle
\thispagestyle{empty}
\pagestyle{empty}

\begin{abstract}

    Many nonlinear dynamical systems exhibit 
    symmetry, affording substantial benefits for control design, observer architecture, and data-driven control. 
    While the classical notion of group invariance enables a cascade decomposition of the system into highly structured subsystems, it demands very rigid structure in the original system. 
    Conversely, much more general notions (\textit{e.g.}, partial symmetry) have been shown to be sufficient for obtaining less-structured decompositions. 
    In this work, we propose a middle ground termed ``weak invariance'', studying diffeomorphisms ({resp.}, vector fields) that are group invariant 
    up to a diffeomorphism of ({resp.}, vector field on) the symmetry group.
    Remarkably, we prove that weak invariance implies that this diffeomorphism of ({resp.}, vector field on) the symmetry group \textit{must} be an automorphism ({resp.}, group linear).
    Additionally, we demonstrate that a vector field is weakly invariant if and only if its flow is weakly invariant, where the associated group linear vector field generates the associated automorphisms.
    Finally, we show that weakly invariant systems admit a cascade decomposition in which the dynamics are group affine along the orbits.
    Weak invariance thus generalizes both classical invariance and the important class of group affine dynamical systems on Lie groups, laying a foundation for new methods of symmetry-informed control and observer design.
\end{abstract}

\section{Introduction}

Group invariance is a powerful property, characteristic of  many dynamical systems of practical significance.
Tracing the modern lineage of such ideas over almost half a century, group invariance has provided a foundation for numerous methods across the analysis of locomotion \cite{Shapere_Wilczek_1989}, the control of networked multi-agent systems \cite{Vasile2018}, and hierarchical control in robotics \cite{Mishra2025}. 
In learning-based control, the inductive bias of group invariance has enabled distributed learning of safe controllers \cite{Bousias2025}, efficient planning algorithms \cite{zhao2023symplan}, and improved generalization in reinforcement learning for robotic manipulation \cite{Wang2022}.
Meanwhile, the estimation community has designed observer architectures for invariant systems evolving on Lie groups in which the estimation errors evolve independently of the system's state \cite{Bonnabel2008tac, Bonnabel2009tac}. More recent work on observer design has exploited invariance for systems evolving on the more general class of homogeneous manifolds \cite{mahony2020equivariantsystemstheoryobserver} to design filters with more favorable linearization properties \cite{Barrau2017,vanGoor2020} and observers with global stability guarantees \cite{vanGoor2021,van2025synchronous,liu2025global}.
Remarkably,  invariant systems admit a decomposition in which one subsystem (on the quotient space) cascades into another subsystem (on the symmetry group) whose dynamics are left-invariant \cite{Grizzle1985}, revealing the fundamental implications of symmetry on system structure.

Classical group invariance is powerful, but it imposes rather rigid requirements on the system dynamics. This
 has motivated the study of relaxed notions of symmetry that capture a broader class of systems.
\textit{Partial symmetry} \cite{nijmeijer1985partial}---a property much more general than invariance---also induces a cascade decomposition (but requires sacrificing any substantive structure in the driven subsystem).
In the Lagrangian setting, \textit{broken symmetry} captures systems that are group invariant up to symmetry-breaking potential forces \cite{Welde2023} or constraints \cite{garcia2024controlled}.
Although early work in the observer community on state-independent error dynamics applied only to left-invariant systems on Lie groups \cite{Bonnabel2008tac, Bonnabel2009tac}, it was shown more recently that the same properties can be attained for the broader class of \textit{group affine} systems \cite{Barrau2017}, and even stronger properties (\textit{i.e.}, the error evolving independently of the inputs) can be achieved even on homogeneous spaces \cite{vanGoor2021,van2025synchronous}.
In reinforcement learning, the reduction of Markov decision processes 
with invariant transition dynamics and rewards is by now well-understood \cite{Wang2022, Panangaden2024}, but more recent work
has demonstrated that such reductions can also be performed even when the transitions are \textit{not} strictly invariant \cite{Zhao2022thesis, welde2025leveragingsymmetryacceleratelearning}.
These disparate examples motivate the need for a balanced notion of symmetry that retains much of the structure enjoyed by group invariant systems, while also accommodating additional systems of practical interest (or 
admitting yet a larger symmetry group for a given system).

In this work, we propose such a balanced notion of symmetry, which we term ``weak invariance''.
For the time being, we restrict our attention to dynamical systems (\textit{i.e.}, systems lacking exogenous control inputs), and study the implications of this relaxed notion of symmetry on the system's structure.
After first reviewing mathematical preliminaries (in Sec.~\ref{sec:mathematical_preliminaries}),  we define a notion of weak invariance for diffeomorphisms (in Sec.~\ref{sec:notions_of_weak_invariance}), taking inspiration from a class of symmetries 
considered in  \cite{welde2025leveragingsymmetryacceleratelearning} for the transition probabilities of a discrete-time stochastic control system. We then extend these ideas to continuous time by defining analogous notions for vector fields and flows and completely characterizing the relationships between them.
A key finding is that our natural definitions for weak invariance implicitly impose a remarkable link with the automorphism group of the symmetry group. Building upon these insights, we demonstrate (in Sec.~\ref{sec:cascade_decompositions}) that a weakly invariant system admits a cascade decomposition in which the driven subsystem is, in fact, group affine. As we discuss (in Sec.~\ref{sec:conclusion}) while concluding the paper, such surprising connections to the important and well-studied class of group affine systems demonstrate the potential of weak invariance as a foundation for symmetry-informed methods suited for yet a broader class of systems.

\section{Mathematical Preliminaries}

\label{sec:mathematical_preliminaries}

We now briefly recall certain essential aspects of differential geometry and Lie group theory.

\subsection{Smooth Manifolds}

Let $M$ be a smooth ($C^\infty$) manifold.
The set of all diffeomorphisms ${f : M \to M}$ is denoted $\Diff(M)$.
The tangent space at each point $p \in M$ is denoted as $\T_p M$ and the tangent bundle is denoted $\T M$.
For a smooth map ${f : M \to N}$ between manifolds, the \textit{differential} of $f$ is a map between tangent bundles denoted ${\diff f : \T M \to \T N}$.
Given a map ${f : X \times Y \to Z}$, we will often use the partial application notation ${f_x = f(x, \, \cdot\, ) : Y \to Z}$ for any fixed ${x \in X}$, and likewise ${f^y = f(\, \cdot\, , y) : X \to Z}$ for any fixed ${y \in Y}$.

A vector field on $M$ is a smooth map ${V : M \to \T M}$ such that ${V(p) \in \T_p M}$ for all ${p \in M}$.
The set of all vector fields on $M$ is denoted by $\mathfrak{X}(M)$.
A flow $\flow$ on $M$ is an $\mathbb{R}$-action on $M$, that is, a smooth map ${\flow : \mathbb{R} \times M \to M}$ satisfying
\begin{align}
    \flow_{t_1 + t_2}(p) &= \flow_{t_1} \circ \flow_{t_2}(p), &
    \flow_0(p) &= p,
    \label{identity_and_compability_of_flows}
\end{align}
for all $t_1, t_2 \in \mathbb{R}$ and $p \in M$.
A vector field $V \in \mathfrak{X}(M)$ is 
\textit{complete}
if there exists an associated flow $\flow^V$ such that
\begin{align}\label{flow_definition}
    \ddt \; \flow^V_t(p) = V \circ \flow^V_t(p),
\end{align}
for any initial condition $p \in M$ and any time $t \in \mathbb{R}$.

\subsection{Lie Groups and Automorphisms}

A Lie group $G$ is a smooth manifold equipped with a smooth group structure.
The product of any two elements ${g,h \in G}$ is written as ${gh \in G}$, and
the identity element of the group is denoted ${e \in G}$.
The left and right translations ${\Left,\Right : G \times G \to G}$ are defined by
$\Left_g(h) := gh$ and $\Right_g(h) := hg$
respectively.
A vector field $V \in \X(G)$ is \textit{left-invariant} if 
${\diff \Left_{g^{-1}} \circ V \circ \Left_g = V}$ for all ${g \in G}$. 
The Lie algebra $\mathfrak{g}$ 
of
the Lie group $G$
consists of all left-invariant vector fields on $G$,
and we make the identification ${\mathfrak{g} \simeq \T_eG}$.

For groups $G$ and $H$, 
a map ${\sigma : G \to H}$ is called a \emph{homomorphism}
if 
it satisfies 
${\sigma(g_1)\sigma(g_2) = \sigma(g_1 g_2)}$
for all ${g_1, g_2 \in G}$.
A homomorphism ${\sigma : G \to G}$ is called an \textit{endomorphism}, and an endomorphism that is also a diffeomorphism is called an \textit{automorphism.}
The set of all automorphisms 
is 
denoted
${\Aut(G) \subset \Diff(G)}$.
When $G$ is connected, $\Aut(G)$
is 
a  Lie group of dimension at most $(\dim G)^2$.
Any
${W \in \mathfrak{X}(G)}$ is called \emph{group linear} if 
for all ${g, h \in G}$, 
\begin{align}
    W(gh) = \diff \Left_g \circ W(h) + \diff \Right_h \circ W(g)
    .
    \label{definition_of_group_linear_vector_field}
\end{align}
The set of all group linear vector fields is written as ${\aut(G) \subset \mathfrak{X}(G)}$. As suggested by the notation, $\aut(G)$ is in fact the Lie algebra of $\Aut(G)$.

\subsection{Group Actions}
A (left) action of a Lie group $G$ on a manifold $M$ is a smooth map ${\Phi : G \times M \to M}$ satisfying\footnote{
A right action  is defined similarly, but with ${\Phi\big(g, \Phi(h, p)\big) = \Phi(hg, p)}$.
}
\begin{align}
    \Phi\big(g, \Phi(h, p)\big) = \Phi(gh, p), &&
    \Phi(e, p) = p,
    \label{group_action_definition}
\end{align}
for all $g,h \in G$ and $p \in M$. 
A group action is said to be: 
\begin{itemize}
    \item \textit{effective}\footnote{Any group action that is \textit{not} effective induces an effective group action with the same orbits by quotienting out its kernel. Thus, we lose nothing by working only with effective group actions, as we will in the remainder.} (or \textit{faithful}) if only the identity element acts trivially (\textit{i.e.}, if ${\Phi_g = \id_M}$, then ${g = e}$), 
    \item \textit{free} if $\Phi_g$ has no fixed points for all ${g \neq e}$ (\textit{i.e.}, if there exists ${p \in M}$ such that ${\Phi_g(p) = p}$, then $g = e$), and
    \item \textit{proper} if the map ${(g,p) \mapsto \big(\Phi_g(p),p\big)}$ is a proper map (\textit{i.e.}, the preimage of any compact set is compact).
\end{itemize}
The \textit{orbit} of any  ${p \in M}$ is given by $\Phi(G,p)$.
The \textit{infinitesimal generator} of $\Phi$ is the map ${(\cdot)_M : \mathfrak{g} \to \mathfrak{X}(M)}$ given by
\begin{equation}
    \xi_M(p) =  \ddt \big(
        \Phi(\flow^{\xi}_t(e), p)
    \big) \big|_{t=0},
\end{equation}
where $\flow^\xi$ is understood as the flow of the left-invariant vector field ${\xi \in \mathfrak{g}}$.
In particular, the left and right translations in $G$ are respectively left and right actions of $G$ on itself.

\subsection{Group Invariance}

The following are standard notions of invariance with respect to a group action ${\Phi : G \times M \to M}$, which we term ``strong'' invariance whenever necessary to distinguish them from weaker notions of symmetry.
In particular, a diffeomorphism ${f : M \to M}$ is  {$\Phi$-invariant} if for all ${g \in G}$,
    \begin{equation}
        f = \Phi_{g^{-1}} \circ f \circ \Phi_{g}.
        \label{strong_phi_invariance_diffeomorphisms}
    \end{equation}
    A 
    vector field ${V \in \X(M)}$ is {$\Phi$-invariant}\footnote{
The condition \eqref{strong_phi_invariance_vector_fields} can also be 
 expressed as ${\diff \Phi_g \circ V = V \circ \Phi_g}$. Since
  $\diff \Phi_g$ is itself a group action,
 strong invariance is a (very) special case of the more general property of  \textit{equivariance}, and some authors simply refer to it as such. We choose our terminology for its consistency with the standard notion of a left-invariant 
 vector field on a Lie group. 
 Indeed, a (strongly) invariant vector field is a fixed (\textit{i.e.}, invariant) point of the $G$-action on the (infinite-dimensional) space $\X(M)$ given by ${(g,V) \mapsto \diff \Phi_{g^{-1}} \circ V \circ \Phi_g}$.
 } if for all ${g \in G}$,
    \begin{equation}
        V = \diff \Phi_{g^{-1}} \circ V \circ \Phi_{g}.
        \label{strong_phi_invariance_vector_fields}
    \end{equation}
    Finally, a flow ${\flow : \mathbb{R} \times M \to M}$ is {$\Phi$-invariant} if the diffeomorphism ${\flow_t : M \to M}$ is $\Phi$-invariant for each ${t \in \mathbb{R}}$. 
The relationship between $\Phi$-invariant vector fields and their flows is well-known 
and completely characterized 
as follows.
\begin{fact}
[{cf. \cite[Thm. 2.3]{Pappas2000}}]
    Let 
    ${\Phi : G \times M \to M}$ 
    be a group action and ${V \in \X(M)}$ be a complete vector field. Then, the flow $\flow^V$ is $\Phi$-invariant if and only if $V$ is $\Phi$-invariant.%
    \label{standard_phi_invariance_if_and_only_if}
\end{fact}

\section{Notions of Weak Invariance}

\label{sec:notions_of_weak_invariance}

In this section, we propose relaxed notions of invariance
with respect to a group action ${\Phi : G \times M \to M}$
for diffeomorphisms, flows, and vector fields.
These definitions are given with respect to additional data defined on the symmetry group, which account for the failure of the entity under consideration to be (strongly) $\Phi$-invariant.
We then study the structure of this additional data, uncovering perhaps unexpected relationships with the automorphism group $\Aut(G)$. We also show that the notions proposed for diffeomorphisms, flows, and vector fields are compatible in the natural way.

\subsection{Weak Invariance of Diffeomorphisms}

We begin by proposing a notion of weak invariance for diffeomorphisms, taking some inspiration from \cite[Eqn. 14]{welde2025leveragingsymmetryacceleratelearning}.%

\begin{definition}
        A diffeomorphism ${f : M \to M}$ is 
        {\textit{weakly $\Phi$-invariant}} if there exists a 
        surjective map
        ${\sigma : G \to G}$ such that for all ${g \in G}$, 
    \begin{equation}
        f = \Phi_{\sigma(g)^{-1}} \circ f \circ \Phi_{g}.
        \label{definition_weak_phi_invariance_of_maps}
    \end{equation}
\end{definition}

Weak invariance is thus a strict generalization of (strong) $\Phi$-invariance  defined in \eqref{strong_phi_invariance_diffeomorphisms}, since we recover 
the traditional notion 
precisely when ${\sigma = \id_G}$. 
Moreover, \eqref{definition_weak_phi_invariance_of_maps} may also be written as ${\Phi_{\sigma(g)} \circ f = f \circ \Phi_g}$, and in fact it can be shown that ${\Upsilon_g := \Phi_{\sigma(g)}}$ is a well-defined group action whenever $\Phi$ is effective. Thus, we see that weak invariance is a special case of the more general notion of \textit{equivariance}, in which truly distinct group actions may act on each side. However, weak invariance is far more specific,
since it
requires the second group action to take the very specific form above. This structure will be  
essential in the results that follow.

Our definition of weak invariance imposes minimal requirements on the map ${\sigma : G \to G}$.
Nonetheless,
our first result shows that such a map enjoys quite rigid structure.

\begin{proposition}
    Let ${\Phi : G \times M \to M}$ be an effective group action and ${f : M \to M}$ be a diffeomorphism that is weakly $\Phi$-invariant with respect to $\sigma : G \to G$. Then, $\sigma \in \Aut(G)$.%
    \label{proposition_sigma_is_an_automorphism}
\end{proposition}

\begin{proof}
    $\Phi$ may also be understood as a group homomorphism
    \begin{equation*}
        \hat{\Phi} : G \to \Diff(M), \ g \mapsto \Phi_g.
    \end{equation*}
    The map $\hat\Phi$ is injective precisely when $\Phi$ is effective, in which case the partial inverse homomorphism
    \begin{equation*}
        \hat\Phi^{-1} : 
        \hat\Phi(G)
        \subseteq \Diff(M) \to G
    \end{equation*}
    is well-defined on the image of $\hat{\Phi}$.
    Rearranging \eqref{definition_weak_phi_invariance_of_maps} to obtain 
    \begin{equation*}
        \Phi_{\sigma(g)}  = f \circ \Phi_{g} \circ f^{-1},
    \end{equation*}
    it is clear that $f \circ \Phi_{g} \circ f^{-1} \in \hat\Phi(G)$, and moreover
    \begin{equation}
        \sigma(g) = \hat{\Phi}^{-1}\big(
        f \circ \Phi_g \circ f^{-1}
        \big)
        \label{definition_of_sigma}
    \end{equation}
    for all ${g \in G}$. 
    The smoothness of $\sigma$ follows from the smoothness of all operations on the right-hand side of \eqref{definition_of_sigma}.
    To show that $\sigma$ is a homomorphism, we use \eqref{definition_of_sigma} to compute
    \begin{align*}
        \sigma(g) \sigma(h) 
        &= 
        \hat{\Phi}^{-1}\big(
        f \circ \Phi_g \circ f^{-1}
        \big)
        \hat{\Phi}^{-1}\big(
        f \circ \Phi_h \circ f^{-1}
        \big)
        \\
        &= 
        \hat{\Phi}^{-1}\big(
        f \circ \Phi_g \circ f^{-1}
        \circ
        f \circ \Phi_h \circ f^{-1}
        \big)
        \\
        &= 
        \hat{\Phi}^{-1}\big(
        f \circ \Phi_{gh} \circ f^{-1}
        \big)
        = \sigma(gh), 
    \end{align*}
    where the second step relies on the fact that $\hat\Phi^{-1}$ is also a homomorphism.
    To show that $\sigma$ is invertible, we claim that
    \begin{align}
        \sigma^{-1}(g) = \hat{\Phi}^{-1}\big(f^{-1} \circ \Phi_g \circ f\big),
        \label{sigma_inverse_definition}
    \end{align}
    where \eqref{definition_weak_phi_invariance_of_maps} and the surjectivity of $\sigma$ imply 
    that 
    ${f^{-1} \circ \Phi_{g} \circ f}$ ${\in \hat\Phi(G)}$, 
    so
     \eqref{sigma_inverse_definition} is well-defined.
     To verify \eqref{sigma_inverse_definition}, compute
    \begin{align*}
        \sigma^{-1}\circ \sigma(g)
        &= \hat{\Phi}^{-1}\big(f^{-1} \circ \Phi_{\hat{\Phi}^{-1}(f \, \circ\,  \Phi_g \, \circ \, f^{-1})} \circ f\big) 
        \\ &
        = \hat{\Phi}^{-1}(f^{-1} \circ f \circ \Phi_g \circ f^{-1} \circ f) 
        \\ &
        = \hat{\Phi}^{-1}(\Phi_g)
        = g,
    \end{align*}
    and likewise 
    $\sigma\circ \sigma^{-1}(g) = g$,
    for all $g \in G$.
\end{proof}

    In effect, \eqref{definition_of_sigma} gives an explicit formula that $\sigma$ must satisfy, which is completely determined\footnote{In fact, if $\Phi$ is effective and $G$ is connected, then any map $\sigma$ satisfying \eqref{definition_weak_phi_invariance_of_maps} 
     is \textit{automatically} surjective, since it must be an injective endomorphism.} by the diffeomorphism $f$ and the group action $\Phi$.
    Thus, it is clear that $\sigma$ is unique (\textit{i.e.},  if $\Phi$ is effective, 
    a diffeomorphism cannot be weakly $\Phi$-invariant with respect to two different automorphisms).

\subsection{Weak Invariance of Flows}

Recall that a flow ${\flow : \mathbb{R} \times M \to M}$ gives a diffeomorphism $\flow_t$ for each fixed time ${t \in \mathbb{R}}$.
We extend our definition of weak invariance of flows in the natural way and establish how the weak invariance of a flow evolves through time.

\begin{definition}
\label{definition_weak_invariance_of_flows}
A flow  ${\flow : \mathbb{R} \times M \to M}$ is \textit{weakly $\Phi$-invariant} if the map $\flow_t : M \to M$  is weakly $\Phi$-invariant for all $t \in \mathbb{R}$.
\end{definition}

To determine whether a flow is weakly $\Phi$-invariant for all time, it suffices to check its weak $\Phi$-invariance on an arbitrarily small open time interval. The following straightforward lemma follows from the group property of flows \eqref{identity_and_compability_of_flows}.%
\begin{lemma}
\label{small_time_weak_invariance_implies_all_time_weak_invariance}
    A flow $\flow$ is weakly $\Phi$-invariant if and only if 
    there exists $\varepsilon > 0$ such that for all $t \in (-\varepsilon, \varepsilon)$, the diffeomorphism $\flow_t$ is weakly $\Phi$-invariant.
\end{lemma}

Note that Definition~\ref{definition_weak_invariance_of_flows} does not explicitly impose any direct relationship between the automorphisms $\sigma_{t_1}$ and $\sigma_{t_2}$ associated with different times $t_1, t_2$ of a weakly $\Phi$-invariant flow.
However, the structure of 
the
flow actually imposes a strong relationship between these automorphisms, as follows.

\begin{proposition}
    Let ${\Phi : G \times M \to M}$ be an effective group action and ${\flow : \mathbb{R} \times M \to M}$ be a weakly $\Phi$-invariant flow with respect to some ${\sigma_t : G \to G}$ for each ${t \in \mathbb{R}}$.
    Then, the map ${\sigma : (t,g) \mapsto \sigma_t(g)}$  is itself a flow on $G$.
    \label{proposition_sigma_t_is_a_flow}
\end{proposition}

\begin{proof} 
    It follows from \eqref{definition_of_sigma} that for each $t \in \mathbb{R}$,
        \begin{equation}
                    \sigma_t(g) = \hat{\Phi}^{-1}\big(
        \flow_t \circ \Phi_g \circ \flow_{-t}
        \big),
        \label{definition_of_sigma_for_arbitrary_flow}
        \end{equation}
        and thus  ${\sigma : \mathbb{R} \times G \to G}$ is smooth, since all operations on the right-hand side of \eqref{definition_of_sigma_for_arbitrary_flow} are smooth with respect to both $g$ and $t$. 
        Considering the definition \eqref{identity_and_compability_of_flows}, we verify 
        \begin{equation*}
            \sigma_0(g) = \hat{\Phi}^{-1}\big(
        \flow_0 \circ \Phi_g \circ \flow_{0}
        \big) 
        =\hat{\Phi}^{-1}\big(
        \Phi_g
        \big) = g, 
        \end{equation*}
        and for any ${t_1}, t_2 \in \mathbb{R}$, we have
        \begin{align*}
            \sigma_{t_1} \circ \sigma_{t_2}(g) &= 
            \hat{\Phi}^{-1}\big(
                \flow_{t_1} \circ \Phi_{
                \hat{\Phi}^{-1}(
            \flow_{t_2} \, \circ \, \Phi_g \, \circ \, \flow_{-{t_2}}
            )
                } \circ \flow_{-{t_1}}
            \big)
            \\ &=
\hat{\Phi}^{-1}\big(
                \flow_{{t_1}+{t_2}} \circ \Phi_g \circ \flow_{-({t_1}+{t_2})}
            \big)
            = \sigma_{{t_1}+{t_2}}(g).
            \end{align*}
            This shows that $\sigma : \mathbb{R} \times  G \to G$ is indeed a flow.
\end{proof}

It is thus
natural to say that a weakly $\Phi$-invariant flow on $M$ is, in particular, weakly $\Phi$-invariant with respect to another flow ${\sigma : \mathbb{R} \times G \to G}$. 
In this sense, $\sigma$ can be interpreted as ``dynamics in the symmetry group'' corresponding to the given dynamics in the state space. In particular, the special case of ``strong'' $\Phi$-invariance corresponds 
to the case in which these dynamics are trivial (\textit{i.e.}, stationary).

\subsection{Weak Invariance of Vector Fields}

The final notion of weak invariance we propose in this paper is for vector fields.
In addition to showing a close relationship with group linear vector fields (\textit{i.e.}, those in $\aut(G)$), we will show that weak invariance of vector fields and flows are indeed naturally and intimately related.

\begin{definition}
A vector field ${V \in \X(M)}$ is \textit{weakly $\Phi$-invariant} if there exists a vector field ${W \in \X(G)}$ such that for all ${g \in G}$ and all ${p \in M}$, 
    \label{def:weak_invariance_of_vector_fields}
    \begin{equation}
        \diff \Phi\big( W(g),  V(p) \big) = V \circ \Phi(g,p).
        \label{definition_weak_phi_invariance_vector_field}
    \end{equation}
\end{definition}

As with our definition of weak invariance for diffeomorphisms, this is a strict generalization of the usual notion of (strong) $\Phi$-invariance given in \eqref{strong_phi_invariance_vector_fields}, which is simply the special case of \eqref{definition_weak_phi_invariance_vector_field} in which $W$ is the zero vector field on $G$. 
Although we have once again imposed no explicit structural assumptions on the vector field $W$, 
we nonetheless obtain the following result, analogous to Proposition~\ref{proposition_sigma_is_an_automorphism}.

\begin{proposition}\label{prop:weak_invariance_implies_group_linear}
    Let ${\Phi : G \times M \to M}$ be an effective group action and ${V \in \X(M)}$ be a vector field that is weakly $\Phi$-invariant with respect to ${W \in \mathfrak{X}(G)}$. Then, $W \in \aut(G)$.
\end{proposition}

\begin{proof}
Since $\Phi$ is effective, the infinitesimal generator map ${(\cdot)_M : \mathfrak{g} \to \X(M)}$ is injective. Moreover, 
${\xi_M(p) = \diff \Phi^p(\xi)}$ for all ${\xi \in \mathfrak{g}}$ and all ${p \in M}$, and ${\diff \Left_{(gh)^{-1}}}$ is a diffeomorphism for all ${g, h \in G}$. Thus, from \eqref{definition_of_group_linear_vector_field},
it will suffice to show that for all $p \in M$ and all $g,h \in G$, 
\begin{equation*}
    \begin{aligned}
     \diff & \Phi^p   \circ \diff \Left_{(gh)^{-1}} \big( W(gh) \big)
     \\ &= 
     \diff \Phi^p \circ
     \diff \Left_{(gh)^{-1}}
     \big( \diff \Left_g \circ W(h) 
     + \diff \Right_h \circ W(g) \big),
\end{aligned} 
\end{equation*}
or equivalently that
\begin{equation}
    \diff \Phi^p \circ W(gh) = \diff \Phi^p  \big( \diff \Left_g \circ W(h) 
     + \diff \Right_h \circ W(g) \big),
     \label{target_to_be_shown_aut_G}
\end{equation} 
where we obtained \eqref{target_to_be_shown_aut_G} using the identity
\begin{equation}
    {\diff \Phi^p \circ \diff \Left_g = \diff \Phi_g \circ \diff \Phi^p}.
    \label{left-translation-identity}
\end{equation}
Then, using the identity $\diff \Phi^p \circ \diff \Right_h = \diff \Phi^{\Phi(h,p)}$, the right-hand side of \eqref{target_to_be_shown_aut_G} can be rewritten as
\begin{equation}
    \begin{aligned}
     \diff  \Phi^p&  \big( \diff \Left_g \circ W(h) 
     + \diff \Right_h \circ W(g) \big) 
     \\ &
         =
     \diff \Phi_g \circ \diff \Phi^p  \circ W(h) 
     + \diff \Phi^{\Phi(h,p)}
     \circ W(g),
     \label{rhs_using_identity}
\end{aligned}
\end{equation}
Since $V$ is weakly $\Phi$-invariant with respect to $W$, we may 
rearrange \eqref{definition_weak_phi_invariance_vector_field}
to show that for all $g \in G$ and any $p \in M$, 
    \begin{align}
     \diff \Phi^p \circ W(g) = V \circ \Phi(g,p) - \diff \Phi_{g} \circ V(p).
     \label{implicit_W_definition}
\end{align}
Using \eqref{implicit_W_definition} to rewrite both terms on the right-hand side of \eqref{rhs_using_identity} and then canceling like terms, we may verify that
\begin{align*}
     \diff \Phi_g & \circ \diff \Phi^p  \circ W(h) 
     + \diff \Phi^{\Phi(h,p)}
     \circ W(g)
     \\ 
       &
       \begin{aligned}
                =\diff \Phi_g  \big( 
     V &\circ \Phi(h,p) - \diff \Phi_{h} \circ V(p)
     \big)
      \ + \\ 
     &\big(
     V \circ  \Phi(gh,p) - \diff \Phi_{g} \circ V \circ \Phi(h,p)
     \big)
       \end{aligned}
    \\    &=
     V \circ \Phi_{gh}(p) - \diff \Phi_{gh} \circ V(p)
     = \diff \Phi^p \circ W(gh),
\end{align*}
where the last equality above relies again on \eqref{implicit_W_definition}.
Having thus verified \eqref{target_to_be_shown_aut_G}, this completes the argument.
\end{proof}

    Just as for weakly $\Phi$-invariant diffeomorphisms,  if $\Phi$ is effective, 
    a given vector field may be weakly $\Phi$-invariant with respect to only a single vector field ${W \in \X(G)}$ (\textit{i.e.}, $W$ is unique). To show this, we note that vector fields on Lie groups are in unique correspondence with their left trivializations, which are given for each $W \in \X(G)$  by 
    \begin{equation}
        \xi^W : G \to \mathfrak{g}, \ g  \mapsto \diff \Left_{g^{-1}} \circ W(g).
        \label{left_trivialized_vector_field}
    \end{equation}
If $V$ is weakly $\Phi$-invariant with respect to $W$, we may re-express \eqref{definition_weak_phi_invariance_vector_field} using \eqref{left-translation-identity} and \eqref{left_trivialized_vector_field} to conclude that
\begin{equation}
    \diff \Phi^p \circ \xi^W(g)  = \diff \Phi_{g^{-1}} \circ V \circ \Phi_g(p) -  V(p).
    \label{almost_in_infintesimal_generator_form}
\end{equation}
Recognizing the left-hand side of \eqref{almost_in_infintesimal_generator_form} as the infinitesimal generator of $\Phi$ applied to ${\xi^W(g) \in \mathfrak{g}}$ yields 
\begin{equation}
    \big(\xi^W(g)\big)_M  = \diff \Phi_{g^{-1}} \circ V \circ \Phi_g -  V. 
    \label{infinitesimal_generator_definition_of_W}
\end{equation}
Since ${(\cdot)_M : \mathfrak{g} \to \X(M)}$ is injective, \eqref{infinitesimal_generator_definition_of_W} uniquely determines $W$ in terms of $V \in \X(M)$ and the group action $\Phi$. 

\begin{remark}
In view of \eqref{infinitesimal_generator_definition_of_W}, it is informative to study the ``residual'' 
$\Delta^V : G \to \X(M)$ of any $V \in \X(M)$, given by 
    \begin{equation}
        \Delta^V : g \mapsto \diff \Phi_{g^{-1}} \circ V \circ \Phi_{g} - V.
    \end{equation}
    This residual detects the failure of $V$ to be (strongly) $\Phi$-invariant at each ${g \in G}$, since 
    $\Delta^V(g)$ is the zero vector field on $M$ for all ${g \in G}$
    precisely when $V$ is (strongly) $\Phi$-invariant.
    Conversely,
    the residual at each ${g \in G}$ 
    of a weakly $\Phi$-invariant vector field
    is given by the infinitesimal generator of $\Phi$ corresponding to $\xi^W(g) \in \mathfrak{g}$, namely,
    \begin{equation}
        \Delta^V(g) = \big(\xi^W(g)\big)_M,
        \label{weak_invariance_residual}
    \end{equation}
    whereas the residual of a vector field that is \textit{not} weakly $\Phi$-invariant will take a more general form. In fact, the notion of partial symmetry proposed in \cite{nijmeijer1985partial} (specialized to the case of vector fields) simply requires this residual to be tangent to the orbits of $\Phi$, whereas our notion of weak invariance requires the residual to take the more structured form \eqref{weak_invariance_residual}.
\end{remark}

We close this section with the following theorem characterizing the relationship between weakly invariant flows and their generating vector fields.

\begin{theorem}
    \label{weak_phi_invariance_if_and_only_if}
    Let 
    ${\Phi : G \times M \to M}$ 
    be an effective group action, ${V \in \X(M)}$ be 
    complete, and $G$ be connected%
    . Then, 
    $V$ is weakly $\Phi$-invariant (with respect to ${W}$)
    if and only if 
    ${\flow^V}$ is weakly $\Phi$-invariant 
    (with respect to $\flow^W$)%
    .
\end{theorem}

\begin{proof}
    To show necessity, we first assume that $V$ is weakly $\Phi$-invariant (as in Def.~\ref{def:weak_invariance_of_vector_fields}) with respect to ${W \in \X(G)}$.
    Considering any $g \in G$ and any $q \in M$, we define 
    a curve 
    \begin{align}
        p_1
        &: 
        \mathbb{R} \to M, \ 
        t \mapsto
        \flow_t^V \circ \Phi(g,q).
        \label{p1_definition}
    \end{align}
    Differentiating $p_1$ and using \eqref{flow_definition}, we obtain
    \begin{equation*}
        \ddt \big(p_1(t)\big) = 
        \ddt \big(
        \flow^V_t \circ \Phi(g,q)
        \big) = V \circ 
        \underbrace{\flow^V_t \circ \Phi(g,q)}_{=p_1(t)},
    \end{equation*}
    and evaluating $p_1$ at $t=0$, we verify that
    \begin{equation*}
         p_1(0) = \flow_0^V \circ \Phi(g,q) = \Phi(g,q).
    \end{equation*}
    Thus,  $p_1$ is a solution to the initial value problem (IVP)
    \begin{equation}
        \dot p(t) = V\big(p(t)\big), \quad p(0) = \Phi(g,q).
        \label{weakly_invariant_IVP}
    \end{equation}
We now define another curve given by
\begin{align}
        \label{p2_definition}
        p_2
        &: 
        \mathbb{R} \to M, \ 
        t \mapsto
        \Phi\big(\flow_t^W(g), \flow_t^V(q) \big),
\end{align}
where we recall that since $G$ is connected, any ${W \in \aut(G)}$ is complete \cite{Jouan_2010}, since it can be lifted to a left-invariant vector field on $\Aut(G)$. 
Differentiating $p_2$, we obtain
    \begin{align*}
        \ddt \big(p_2(t)\big) &= 
        \ddt \Big(
        \Phi\big(\flow_t^W(g), \flow_t^V(q) \big)
        \Big) 
        \\ &= 
        \diff \Phi \Big(
        \underbrace{\ddt \flow_t^W(g)}_{= W \circ \flow^W_t(g)}, \underbrace{\ddt \flow_t^V(q)}_{= V \circ \flow^V_t(q)}
        \Big)
        \\ &= 
        V \circ \underbrace{\Phi\big(
\flow^W_t(g),\flow^V_t(q)
        \big)}_{=p_2(t)},
    \end{align*}
    where the last step relies on the weak $\Phi$-invariance of $V$.
    Evaluating $p_2$ at $t=0$, we verify that
    \begin{equation*}
        p_2(0) = \Phi\big(\flow^W_0(g),\flow_0^V(q)\big) = \Phi(g,q),
    \end{equation*}
   showing that $p_2$ is also a solution to the same IVP \eqref{weakly_invariant_IVP}. 
    Since initial value problems have unique solutions, 
    ${p_1 = p_2}$. In view of 
    \eqref{p1_definition}, \eqref{p2_definition}, and
    \eqref{definition_weak_phi_invariance_of_maps},
    this shows necessity.
    
    To show sufficiency, we now assume that
        $\flow^V$ is weakly $\Phi$-invariant with respect to ${\flow^W}$ (as in Def.~\ref{definition_weak_invariance_of_flows}).
       This 
       means that
       for any ${g \in G}$, ${q \in M}$, and $t \in \mathbb{R}$, 
       \begin{equation*}
           \Phi \big(\flow^W_t(g), \flow^V_t(q)\big)
           =
           \flow^V_t \circ \Phi(g,q).
       \end{equation*}
       Differentiating both sides with respect to $t$ (using the chain rule on the left-hand side) and applying \eqref{flow_definition} yields
        \begin{equation*}
            \diff \Phi \big(
            W \circ \flow^W_t(g)
            ,
            V \circ \flow^V_t(q)\big)
            = 
            V \circ \flow^V_t \circ \Phi(g,q).
        \end{equation*}
        Evaluating this equality at $t = 0$ 
        yields
        precisely \eqref{definition_weak_phi_invariance_vector_field}%
        .
\end{proof}

    The assumption that $G$ is connected can be relaxed considerably (\textit{cf.} \cite{Hochschild1952_automorphism}). Also,
    Fact~\ref{standard_phi_invariance_if_and_only_if} (relating strongly $\Phi$-invariant vector fields and their flows)
    can be recovered as a 
    corollary of Theorem~\ref{standard_phi_invariance_if_and_only_if} by taking $V$ to be weakly $\Phi$-invariant with respect to 
    ${0_{\X(G)} \in \aut(G)}$
    and $\flow^V_t$ to be weakly $\Phi$-invariant with respect to ${\id_G \in \Aut(G)}$ for all $t \in \mathbb{R}$.

\section{Cascade Decompositions of \\ Weakly Invariant Dynamical Systems}

\label{sec:cascade_decompositions}

In this section, we consider ${\Phi : G \times M \to M}$ to be a free and proper group action. Thus,
the projection ${\pi : M \to M/G}$, which maps each point ${p \in M}$ to its orbit $\Phi_G(p)$, is a surjective submersion.
In this setting, we show that any weakly invariant vector field can be decomposed into an independent vector field on the quotient manifold $M / G$, which cascades into highly structured dynamics along the orbits.
The following lemma is well-known---in fact, it is a special case of \cite[Prop. 2.7]{nijmeijer1985partial} (see also \cite[Prop. 3.2]{Absil2007}). 

\begin{lemma}\label{lem:induced_vector_field}
Let ${V \in \mathfrak{X}(M)}$ be weakly $\Phi$-invariant%
.
Then $V$ induces a vector field ${\tilde{V} \in \mathfrak{X}(M/G)}$ 
such that for all $p \in M$,
\begin{align}
    \tilde{V}\big(\pi(p)\big) = \diff \pi \circ V(p).
    \label{reduced_Vector_field_definition}
\end{align}
\end{lemma}

\begin{proof}
It suffices to show that ${\diff \pi \circ V(p_1) = \diff \pi \circ V(p_2)}$ whenever ${\pi(p_1) = \pi(p_2)}$ (\textit{i.e.}, whenever 
there exists $g \in G$ such that ${\Phi_g(p_1) = p_2}$).
Since $V$ is weakly $\Phi$-invariant,
\begin{align*}
    V(p_2) 
    &= V \circ \Phi_g(p_1) 
    = \diff \Phi_g \circ V(p_1) + \diff \Phi^{p_1} \circ W(g).
\end{align*}
Since the second term here is tangent to the orbit of $p_1$, it follows that  $\diff \pi \circ V(p_2) = \diff \pi \circ V(p_1)$, as required.
\end{proof}

While the previous lemma characterizes the system's ``reduced dynamics'' in the quotient space (\textit{i.e.}, from orbit to orbit), our next result describes its evolution 
\textit{along} the orbits.

\begin{theorem}\label{thm:decomposition}
Let ${V \in \mathfrak{X}(M)}$ be weakly $\Phi$-invariant with respect to ${W \in \aut(G)}$, and suppose that 
$\pi$
admits a global section ${s : M/G \to M}$.
If ${p(t) \in M}$ is a curve satisfying ${\dot{p} = V(p)}$, then its decomposition ${p = \Phi_g \circ s(y)}$ 
 evolves as
\begin{subequations}
    \begin{align}
    \dot{y} &= \tilde{V}(y), \label{cascade_quotient} \\
    \dot{g} &= W(g) + \diff \Left_g \circ \hat{V}(y) \label{cascade_orbits},
\end{align}
\end{subequations}
where the map $\hat{V} : M/G \to \mathfrak{g}$ is defined by
\begin{align}
    \hat{V}(y) &= \big(\diff \Phi^{s(y)}\big)^{-1} \big(V \circ s(y) -  \diff s \circ \tilde{V}(y)\big).
    \label{forcing_term_group_affine_along_orbits}
\end{align}
\end{theorem}

\begin{proof}[Proof]
It follows immediately from Lemma~\ref{lem:induced_vector_field} that ${\dot{y} = \tilde{V}(y)}$. 
It remains to derive the dynamics of $g$ described in \eqref{cascade_orbits}.
Substituting the decomposition ${p = \Phi_g \circ s(y)}$ into $\dot{p} = V(p)$, we have
\begin{equation*}
    \ddt \Phi\big(g,s(y)\big) = V \circ \Phi_g \circ s(y).
\end{equation*}
Applying the chain rule on the left-hand side and using the weak $\Phi$-invariance of $V$ on the right-hand side yields
\begin{equation*}
    \diff \Phi \big(\dot{g}, \diff s (\dot{y})\big)
    = 
    \diff \Phi \big(W(g), V \circ s(y)\big).
\end{equation*}
Expanding both sides and rearranging terms, we obtain
\begin{equation*}
    \diff \Phi^{s(y)} \big(\dot{g} - W(g) \big)
    =
    \diff \Phi_g \big( V \circ s(y) -  \diff s (\dot{y})\big).
\end{equation*}
Since $\Phi$ is free, the orbit map $ {\Phi^p : G \to M}$ is an injective immersion for any ${p \in M}$, and therefore $\diff \Phi^{s(y)}$ is invertible over its image for any ${y \in M/G}$. 
Thus,
\begin{equation*}
    \dot{g}
    =
    W(g) + \big(\diff \Phi^{s(y)} \big)^{-1} \circ \diff \Phi_g \big( V \circ s(y) -  \diff s (\dot{y})\big).    
\end{equation*}
Finally, it follows from \eqref{left-translation-identity} and the fact that ${\dot{y} = \tilde{V}(y)}$ that 
\begin{equation*}
    \dot{g}
    =
    W(g) + \diff \Left_g \circ \big(\diff \Phi^{s(y)} \big)^{-1}  \big( V \circ s(y) -  \diff s \circ \tilde{V} (y) \big). \qedhere
\end{equation*} 
\end{proof}

\begin{remark}
    Within the setting of dynamical systems (lacking inputs), the above result establishes a middle ground between the cascade decompositions previously obtained in the more general setting of partial symmetry \cite[Sec.~2]{nijmeijer1985partial} 
    and the less general setting of (strong) invariance \cite[Thm. 4.2]{Grizzle1985}.
    In particular, specializing those prior results (formulated in the more general setting of systems with inputs) to dynamical systems, \cite[Sec.~2]{nijmeijer1985partial} would replace
    \eqref{cascade_orbits} with ${\dot{g} = h(g,y)}$, thus sacrificing any structure besides that of the cascade. On the other hand, Theorem~\ref{thm:decomposition} and \eqref{cascade_orbits} still apply (in particular) to the situation studied in \cite[Thm. 4.2]{Grizzle1985} (since strong invariance is simply the special case of weak invariance with ${W \equiv 0}$), in which case \eqref{cascade_orbits} becomes a left-invariant vector field for each (fixed) value of ${y \in M/G}$.
\end{remark}

\subsection{Group Affine Dynamical Systems}

An important corollary of Theorem~\ref{thm:decomposition} pertains to vector fields on a Lie group $G$.
Specifically, we show that the class of weakly left-invariant dynamical systems on Lie groups is exactly the class of group affine dynamical systems.
However, we stress that weak invariance is much more general; while group affine systems are strictly defined on Lie groups, weak invariance is defined for systems on arbitrary manifolds and does not require a transitive group action.

As introduced in \cite{Barrau2017}, a vector field  ${V \in \X(G)}$ on a Lie group $G$ is called \emph{group affine} if for all $g,h \in G$, it satisfies
\begin{align}
\begin{aligned}
        V(gh) = \ \ &\\ \diff \Left_g \circ  V&(h) + \diff \Right_h \circ V(g) - \diff \Left_g \circ \diff \Right_h \circ V(e).
        \label{group_affine_vector_field}
\end{aligned}
\end{align}

\begin{corollary}
A vector field $V \in \X(G)$ is group affine if and only if it is weakly $\Left$-invariant.
\end{corollary}

\begin{proof}
For the forward direction, we use the fact \cite[Thm. 4.3]{vanGoor2021} that any group affine vector field can be written as
\begin{align*}
    V(g) = W(g) + \diff \Left_g (U)
\end{align*}
for unique choices of ${W \in \aut(G)}$ and ${U \in \mathfrak{g}}$.
Weak $\Left$-invariance is then verified by checking the definition \eqref{definition_weak_phi_invariance_vector_field}:
\begin{align*}
    V \circ \Left_g(h) 
    &= V(gh) 
    = 
    W(gh) + \diff \Left_{gh} (U) \\
    &= \diff \Left_g \circ W(h) + \diff \Right_h \circ W(g) + \diff \Left_g \circ \diff \Left_{h} (U) \\
    &= \diff \Left^h \circ W(g) + \diff \Left_g \big(W(h) + \diff \Left_{h} (U)\big) \\
    &= \diff \Left\big( W(g), V(h)\big),
\end{align*}
as required, where the third equality relies on 
\eqref{definition_of_group_linear_vector_field}.

For the reverse direction, let ${V \in \mathfrak{X}(G)}$ be weakly $\Left$-invariant with respect to ${W \in \aut(G)}$.
Since $\Left$ is transitive, the quotient space $G / G$ is a zero-dimensional singleton.
Therefore, $\pi$ is a constant map, and \eqref{reduced_Vector_field_definition} implies that ${\tilde{V} \equiv 0}$ (\textit{i.e.}, the dynamics in the quotient space are trivial).
Moreover, ${U := \hat{V}(y) \in \mathfrak{g}}$ is a constant (there is only one element ${y \in G/G}$).  Applying Theorem~\ref{thm:decomposition} yields
\begin{align*}
    \dot{g} &= W(g) + \diff \Left_g \circ \hat{V}(y) = W(g) + \diff \Left_g \circ U,
\end{align*}
exactly the form of a group affine vector field from \cite{vanGoor2021}.
\end{proof}

\section{Conclusion}

\label{sec:conclusion}

In conclusion, we have established the notion of weak invariance for dynamical systems, identifying a fertile middle ground between classical (strong) invariance and the more general notion of partial symmetry. Our weak notion of symmetry is defined in terms of additional data on the symmetry group, 
which we proved
must lie within the automorphism group of the symmetry group (or the Lie algebra of the automorphism group). 
We also proved that a flow is weakly invariant if and only if its generating vector field is also, emphasizing the naturality of our definitions. 
Finally, we showed that a weakly invariant dynamical system admits a powerful cascade decomposition in which the driven subsystem is group affine.
Future work will extend these relaxed notions of symmetry to systems with control inputs and exploit the structure retained even in this more general setting for 
control and observer design.

\bibliographystyle{IEEEtran}
\bibliography{references.bib}

\begin{thebibliography}{10}
\providecommand{\url}[1]{#1}
\csname url@samestyle\endcsname
\providecommand{\newblock}{\relax}
\providecommand{\bibinfo}[2]{#2}
\providecommand{\BIBentrySTDinterwordspacing}{\spaceskip=0pt\relax}
\providecommand{\BIBentryALTinterwordstretchfactor}{4}
\providecommand{\BIBentryALTinterwordspacing}{\spaceskip=\fontdimen2\font plus
\BIBentryALTinterwordstretchfactor\fontdimen3\font minus
  \fontdimen4\font\relax}
\providecommand{\BIBforeignlanguage}[2]{{%
\expandafter\ifx\csname l@#1\endcsname\relax
\typeout{** WARNING: IEEEtran.bst: No hyphenation pattern has been}%
\typeout{** loaded for the language `#1'. Using the pattern for}%
\typeout{** the default language instead.}%
\else
\language=\csname l@#1\endcsname
\fi
#2}}
\providecommand{\BIBdecl}{\relax}
\BIBdecl

\bibitem{Shapere_Wilczek_1989}
A.~Shapere and F.~Wilczek, ``Geometry of self-propulsion at low {R}eynolds
  number,'' \emph{Journal of Fluid Mechanics}, vol. 198, p. 557––585, 1989.

\bibitem{Vasile2018}
C.-I. Vasile, M.~Schwager, and C.~Belta, ``{Translational and Rotational
  Invariance in Networked Dynamical Systems},'' \emph{IEEE Transactions on
  Control of Network Systems}, vol.~5, no.~3, pp. 822--832, 2018.

\bibitem{Mishra2025}
H.~Mishra, M.~De~Stefano, and C.~Ott, ``{Is There a Closed-Loop Lagrangian for
  Hierarchical Motion Control?}'' \emph{IEEE Control Systems Letters}, vol.~9,
  pp. 1748--1753, 2025.

\bibitem{Bousias2025}
N.~Bousias, L.~Lindemann, and G.~Pappas, ``{Deep Equivariant Multi-Agent
  Control Barrier Functions},'' in \emph{IEEE/RSJ International Conference on
  Intelligent Robots and Systems}, 2025, pp. 997--1004.

\bibitem{zhao2023symplan}
L.~Zhao, X.~Zhu, L.~Kong, R.~Walters, and L.~L. Wong, ``{Integrating Symmetry
  into Differentiable Planning with Steerable Convolutions},'' in
  \emph{International Conference on Learning Representations}, 2023.

\bibitem{Wang2022}
D.~Wang, R.~Walters, and R.~Platt, ``{{$\mathrm{SO}(2)$}-Equivariant
  Reinforcement Learning},'' in \emph{International Conference on Learning
  Representations}, 2022.

\bibitem{Bonnabel2008tac}
S.~Bonnabel, P.~Martin, and P.~Rouchon, ``{Symmetry-Preserving Observers},''
  \emph{IEEE Transactions on Automatic Control}, vol.~53, no.~11, pp.
  2514--2526, 2008.

\bibitem{Bonnabel2009tac}
------, ``{Non-Linear Symmetry-Preserving Observers on Lie Groups},''
  \emph{IEEE Transactions on Automatic Control}, vol.~54, no.~7, 2009.

\bibitem{mahony2020equivariantsystemstheoryobserver}
R.~Mahony, T.~Hamel, and J.~Trumpf, ``{Equivariant Systems Theory and Observer
  Design},'' \emph{arXiv preprint arXiv:2006.08276}, 2020.

\bibitem{Barrau2017}
A.~Barrau and S.~Bonnabel, ``{The Invariant Extended Kalman Filter as a Stable
  Observer},'' \emph{IEEE Transactions on Automatic Control}, vol.~62, no.~4,
  pp. 1797--1812, 2017.

\bibitem{vanGoor2020}
P.~van Goor, T.~Hamel, and R.~Mahony, ``{Equivariant Filter (EqF): A General
  Filter Design for Systems on Homogeneous Spaces},'' in \emph{IEEE Conference
  on Decision and Control}, 2020, pp. 5401--5408.

\bibitem{vanGoor2021}
P.~van Goor and R.~Mahony, ``{Autonomous Error and Constructive Observer Design
  for Group Affine Systems},'' in \emph{IEEE Conference on Decision and
  Control}, 2021, pp. 4730--4737.

\bibitem{van2025synchronous}
------, ``Synchronous models and fundamental systems in observer design,''
  \emph{IEEE Transactions on Automatic Control}, pp. 1--15, 2026.

\bibitem{liu2025global}
C.~Liu and Y.~Shen, ``{Global Observer Design for a Class of Linear Observed
  Systems on Groups},'' \emph{arXiv:2507.18493}, 2025.

\bibitem{Grizzle1985}
J.~Grizzle and S.~Marcus, ``{The Structure of Nonlinear Control Systems
  Possessing Symmetries},'' \emph{IEEE Transactions on Automatic Control},
  vol.~30, no.~3, pp. 248--258, 1985.

\bibitem{nijmeijer1985partial}
H.~Nijmeijer and A.~Van~der Schaft, ``{Partial Symmetries for Nonlinear
  Systems},'' \emph{Mathematical Systems Theory}, vol.~18, no.~1, pp. 79--96,
  1985.

\bibitem{Welde2023}
J.~Welde, M.~D. Kvalheim, and V.~Kumar, ``{The Role of Symmetry in Constructing
  Geometric Flat Outputs for Free-Flying Robotic Systems},'' in \emph{IEEE
  International Conference on Robotics and Automation}, 2023, pp.
  12\,247--12\,253.

\bibitem{garcia2024controlled}
J.~S. Garcia and T.~Ohsawa, ``{Controlled Lagrangians and Stabilization of
  Euler--Poincar{\'e} Equations with Symmetry Breaking Nonholonomic
  Constraints},'' \emph{Journal of Nonlinear Science}, vol.~34, no.~5, p.~91,
  2024.

\bibitem{Panangaden2024}
P.~Panangaden, S.~Rezaei-Shoshtari, R.~Zhao, D.~Meger, and D.~Precup, ``{Policy
  Gradient Methods in the Presence of Symmetries and State Abstractions},''
  \emph{Journal of Machine Learning Research}, vol.~25, no.~71, pp. 1--57,
  2024.

\bibitem{Zhao2022thesis}
R.~Y. Zhao, ``\BIBforeignlanguage{English}{{Continuous Homomorphisms and
  Leveraging Symmetries in Policy Gradient Algorithms for Markov Decision
  Processes}},'' Master's thesis, McGill University, 2022.

\bibitem{welde2025leveragingsymmetryacceleratelearning}
J.~Welde, N.~Rao, P.~Kunapuli, D.~Jayaraman, and V.~Kumar, ``{Leveraging
  Symmetry to Accelerate Learning of Trajectory Tracking Controllers for
  Free-Flying Robotic Systems},'' in \emph{IEEE International Conference on
  Robotics and Automation}, 2025, pp. 15\,121--15\,127.

\bibitem{Pappas2000}
G.~Pappas, G.~Lafferriere, and S.~Sastry, ``{Hierarchically Consistent Control
  Systems},'' \emph{IEEE Transactions on Automatic Control}, vol.~45, no.~6,
  pp. 1144--1160, 2000.

\bibitem{Jouan_2010}
P.~Jouan, ``{Equivalence of control systems with linear systems on Lie groups
  and homogeneous spaces},'' \emph{ESAIM: Control, Optimisation and Calculus of
  Variations}, vol.~16, no.~4, p. 956–973, 2010.

\bibitem{Hochschild1952_automorphism}
G.~Hochschild, ``{The Automorphism Group of a Lie Group},'' \emph{Transactions
  of the American Mathematical Society}, vol.~72, no.~2, pp. 209--216, 1952.

\bibitem{Absil2007}
P.-A. Absil, C.~Lageman, and J.~H. Manton, ``Design of continuous-time flows on
  intertwined orbit spaces,'' in \emph{IEEE Conference on Decision and
  Control}, 2007, pp. 6244--6249.

\end{thebibliography}

\appendix

In this appendix, we provide an explicit proof of 
Lemma~\ref{small_time_weak_invariance_implies_all_time_weak_invariance}, which states that weak invariance of a flow locally in time implies weak invariance globally in time.
For any 
map ${h \in \Diff(M)}$
and 
any integer 
$n \in \mathbb{Z}$, the 
map ${h^n \in \Diff(M)}$
is defined by the (two-sided) recursion
\begin{align}
    h^{n+1} = h^n \circ h, && h^0 = \id_M,
\end{align}
where $\id_M : M \to M$ denotes the identity map.

\begin{proof}[Proof of Lemma~\ref{small_time_weak_invariance_implies_all_time_weak_invariance}]
    Necessity is automatic. For sufficiency, 
    fix any ${t \in \mathbb{R}}$.
    Then, there exists  ${\delta \in (-\varepsilon, \varepsilon)}$ and ${n \in \mathbb{N}}$ such that ${t = n\delta}$. Thus, by the group property of flows,
    \begin{equation*}
        \flow_t = \flow_{n\delta} = \underbrace{\flow_\delta \circ \cdots \circ \flow_\delta}_{n \textrm{ times}} = (\flow_\delta)^n.
    \end{equation*}
    It follows that
    \begin{align*}
    \nonumber
        \flow_t &\circ \Phi_g \circ \flow_{t}^{-1} 
        \\ &= 
        (\flow_{\delta})^n \circ \Phi_g \circ (\flow_{\delta})^{-n} 
        \\
        &= 
        (\flow_{\delta})^{n-1} \circ 
        \underbrace{\flow_{\delta} \circ \Phi_g \circ \flow^{-1}_{\delta}}_{=\Phi_{\sigma_\delta(g)}}
        \circ (\flow_{\delta})^{-(n-1)},
    \end{align*}
    where the equality underneath follows from the weak $\Phi$-invariance of $\flow$ locally in time, since ${\delta \in (-\varepsilon, \varepsilon)}$. Applying this property ${n-1}$ additional times, we verify that 
    \begin{align*}
        \flow_t \circ \Phi_g \circ \flow_{t}^{-1} = \Phi_{(\sigma_\delta)^n(g)},
    \end{align*}
    and thus $\flow_t$ is weakly $\Phi$-invariant with respect to ${\sigma_t : g \mapsto (\sigma_{\delta})^n(g)}$. Since ${t \in \mathbb{R}}$ was taken to be arbitrary, we have shown that $\flow$ is weakly $\Phi$-invariant globally in time.
\end{proof}

\end{document}